\documentclass[english]{elsarticle}
\usepackage[T1]{fontenc}
\usepackage[latin9]{inputenc}
\usepackage{float}
\usepackage{amsmath}
\usepackage{amsthm}
\usepackage{amssymb}
\usepackage[all]{xy}

\makeatletter

\providecommand{\tabularnewline}{\\}

\numberwithin{equation}{section}
\numberwithin{figure}{section}
\theoremstyle{plain}
\newtheorem{thm}{\protect\theoremname}
  \theoremstyle{plain}
  \newtheorem{lem}[thm]{\protect\lemmaname}
  \theoremstyle{plain}
  \newtheorem{cor}[thm]{\protect\corollaryname}
  \theoremstyle{definition}
  \newtheorem{example}[thm]{\protect\examplename}

\journal{Journal of Geometry and Physics}

\usepackage[all]{xy}
\usepackage{amsmath}
\usepackage{slashed}
\usepackage{graphicx}
\usepackage{amssymb}
\usepackage{hyperref}

\hypersetup{colorlinks=true, linkcolor=blue, citecolor=blue, filecolor=blue, pagecolor=blue, urlcolor=red, pdfauthor={Yuri Ximenes Martins}, pdftitle={Topological and Geometric Obstructions on Einstein-Hilbert-Palatini Theories}}
\renewcommand{\textendash}{--}

\makeatother

\usepackage{babel}
  \providecommand{\corollaryname}{Corollary}
  \providecommand{\examplename}{Example}
  \providecommand{\lemmaname}{Lemma}
\providecommand{\theoremname}{Theorem}

\begin{document}
\begin{frontmatter}

\title{Topological and Geometric Obstructions on Einstein-Hilbert-Palatini
Theories}

\author{Yuri Ximenes Martins\corref{cor1}\fnref{label2}}
\author{Rodney Josu\'{e} Biezuner,\fnref{label2}}
\address[label2]{Departamento de Matem\'{a}tica, ICEx, Universidade Federal de Minas Gerais, Av. Ant\^{o}nio Carlos 6627, Pampulha, CP 702, CEP 31270-901, Belo Horizonte, MG, Brazil}
\begin{abstract}
In this article we introduce $A$-valued Einstein-Hilbert-Palatini
functional ($A$\textbf{-}EHP) over a $n$-manifold $M$, where $A$
is an arbitrary graded algebra, as a generalization of the functional
arising in the study of the first order formulation of gravity. We
show that if $A$ is weak $(k,s)$-solvable, then $A$-EHP is non-null
only if $n<k+s+3$. We prove that essentially all algebras modeling
classical geometries (except semi-Riemannian geometries with specific
signatures) satisfy this condition for $k=1$ and $s=2$, including
Hitchin's generalized complex geometry, Pantilie's generalized quaternionic
geometries and all other generalized Cayley-Dickson geometries. We
also prove that if $A$ is concrete in some sense, then a torsionless
version of $A$-EHP is non-null only if $M$ is K\"{a}hler of dimension
$n=2,4$. We present our results as obstructions to $M$ being an
Einstein manifold relative to geometries other than semi-Riemannian.
\end{abstract}
\end{frontmatter}

\section{Introduction}

The standard theory of gravity is General Relativity \citep{wald,hawking_ellis},
which is formulated in Lorentzian geometry: spacetime is regarded
as a Lorentzian manifold $(M,g)$ whose Lorentzian metric $g$ is
a critical point of the Einstein\textendash Hilbert action functional
\begin{equation} S_{EH}[g]=\int_{M}(R_{g}-2\Lambda)\cdot\omega_{g},\label{einstein_hilbert_action} \end{equation}defined
on the space of all possible Lorentzian metrics in $M$, $\Lambda\in\mathbb{R}$
is the cosmological constant and $R_{g}$ and $\omega_{g}$ are, respectively,
the scalar curvature and the volume form of $g$. The Euler-Lagrange
equations of $S_{EH}$ are the Einstein's equations
\begin{equation}
\operatorname{Ric}_{g}-\frac{1}{2}R_{g}g+\Lambda g=0,\label{einstein_eq}
\end{equation}
which are equivalent to saying that $(M,g)$ is an Einstein manifold,
i.e, $\operatorname{Ric}_{g}=\frac{2}{n-2}\Lambda g$. The problem
of determining which manifolds admit an Einstein structure is old
and a general classification of them remains an open problem \citep{einstein_manifolds_BESSE,review_einstein}. 

The action functional $S_{EH}$ can be considered not only in Lorentzian
metrics, but also in metrics of arbitrary signature. The Euler-Lagrange
equations have the same form and they are also equivalent to the Einstein
manifold condition. Despite the similarity, being Einstein depends
strongly on the signature (and, therefore, on the underlying geometry).
For instance, if $M$ is compact, orientable, non-parallelizable and
K\"{a}hler, then it does not admit any Lorentzian metric\footnote{In fact, K\"{a}hler condition requires $b_{i}(M)=0$ for $i$ odd,
so that $\chi(M)=\sum_{i}b_{2i}(M)$, which is positive, since $M$
is non-parallelizable. On the other hand, compact manifolds admit
Lorentzian metrics only if $\chi(M)=0$.}, while Calabi-Yau manifolds, hyperK\"{a}hler and quaternion K\"{a}hler
manifolds are the basic examples of Riemannian Einstein manifolds.

Due to the difficulty of classifying semi-Riemmanian Einstein-manifolds,
many generalizations and closely related concepts were introduced
and studied, such as generalized Einstein manifolds \citep{generalized_Einstein},
quasi-Einstein manifolds \citep{quasi_einstein}, generalized quasi-Einstein
manifolds \citep{generalized_quasi_Einstein}, super quasi-Einstein
manifolds \citep{super_quasi_generalized}, and many others. Each
of these generalizations concerns semi-Riemannian manifolds. Here
we will consider another generalization: the problem of finding Einstein
manifolds for geometries other than semi-Riemannian. Because semi-Riemannian
Einstein manifolds are critical points of the Einstein-Hilbert functional,
the idea is to define Einstein manifolds relative to other geometries
as critical points of functionals which are analogues of $S_{EH}$
in those geometries. The classical examples of geometries can be described
by tensors $t$ on $M$ fulfilling additional integrability conditions.
Considering connections on $M$ compatible with $t$ (in the sense
that $\nabla t=0$) we can take their curvature tensor and Ricci curvature.
In order to write down the analogous action we need to contract the
Ricci curvature with the tensor $t$, getting the \emph{scalar curvature
}$R_{t}$ of $t$. If $M$ is oriented we can define an action functional
$S[t]=\int(R_{t}-2\Lambda)\cdot\omega$ on the moduli space of those
tensors. This makes sense only when $t$ has rank two. For example,
we can write a direct analogue of General Relativity in symplectic
geometry, with action functional $S[\omega]=\int_{M}(\pi^{ij}\operatorname{Ric}_{ij}-2\Lambda)\cdot\omega^{n}$
on a $2n$-dimensional manifold, where $\pi^{ij}$ are the components
of the bivector induced by $\omega$. But the tensor $\operatorname{Ric}_{ij}$
is always symmetric, so that $\pi^{ij}\operatorname{Ric}_{ij}=0$
and the action therefore is trivial \citep{symplectic_curvature_1,symplectic_curvature_2}.
Also, there is no obvious analogue of $S_{EH}$ (and, consequently,
of Einstein manifolds) relative to K\"{a}hler, quaternion K\"{a}hler
and hyper-K\"{a}hler geometry. This fact should not be viewed as
a contradiction to the above cited fact that Calabi-Yau manifolds,
quaternion K\"{a}hler and hyper-K\"{a}hler manifolds are Riemmanian\textbf{
}Einstein manifolds nor to the applications of those geometries in
classical (i.e, semi-Riemannian) General Relativity, specially in
the canonical formalism \citep{ashteckar,complex_GR}.

We recall, however, that Einstein's equations can be obtained as critical
points of another functional, usually known as the\emph{ Einstein\textendash Hilbert\textendash Palatini}
(EHP) \emph{action} \emph{functional}, which arises in the so-called
\emph{first order formulation of gravity} (also called\emph{ tetradic
gravity}). Instead of a semi-Riemannian metric, this functional is
defined on the space of reductive Cartan connections on the frame
bundle $FM\rightarrow M$, relative to the group reduction $O(r,s)\hookrightarrow\mathbb{R}^{r,s}\rtimes O(r,s)$.
These can be identified with pairs of $1$-forms $e$ and $\omega$
in $FM$, called \emph{tetrad} and \emph{spin connection}, with values
in $\mathbb{R}^{n-1,1}$ and $\mathfrak{o}(n-1,1)$, respectively.
It is usual to assume that $e$ is pointwise an isomorphism. The action
itself is given by \begin{equation} S_{EHP}[e,\omega]=\int_{M}\operatorname{tr}(\curlywedge_{\rtimes}^{n-2}e\curlywedge_{\rtimes}\Omega+\frac{\Lambda}{(n-1)!}\curlywedge_{\rtimes}^ne),\label{EHP_action} \end{equation}
where $\curlywedge_{\rtimes}$ is a type of ``wedge product\footnote{In this paper we will work with many different types of wedge products,
satisfying very different properties. Therefore, in order to avoid
confusion, we will not follow the literature, but introduce specific
symbols for each of them.}'' induced by matrix multiplication in $O(r,s)$, $\curlywedge_{\rtimes}^{k}\alpha=\alpha\curlywedge_{\rtimes}...\curlywedge_{\rtimes}\alpha$
and $\Omega=d\omega+\omega\curlywedge_{\rtimes}\omega$ is the curvature
of $\omega$. It is a well-known fact \citep{baez_knots,PhD_cartan_connections}
that in dimension $n=4$ and Lorentzian signature varying (\ref{EHP_action})
in relation to $e$ and $\omega$ we get 
\begin{equation}
e\curlywedge_{\rtimes}\Omega+\frac{\Lambda}{(n-1)!}e\curlywedge_{\rtimes}e\curlywedge_{\rtimes}e=0\quad\text{and}\quad e\curlywedge_{\rtimes}\Theta=0,\label{eintein_palatini_equations}
\end{equation}
where $\Theta=d\omega+\omega\curlywedge_{\rtimes}e$ is the \emph{torsion
of $\omega$. }Lorentz metrics on $M$ can be identified as pullbacks
of the canonical Minkowski metric on $\mathbb{R}^{r,s}$ via tetrad
and the connections $\omega$ are compatible with them. Furthermore,
since $e$ is pointwise an isomorphism, the second equation in (\ref{eintein_palatini_equations})
is equivalent to the torsion-free condition $\Theta=0$, so that $\omega$
is necessarily a Levi-Civita connection. Via these identifications,
the first equation in (\ref{eintein_palatini_equations}) is just
Einstein's equation. Therefore, a $4$-dimensional Einstein manifold
can also be regarded as a critical point (relative to $e$) of $S_{EHP}$. 

The advantage of this new approach is that unlike $S_{EH}$, the EHP
functional $S_{EHP}$ (and, consequently, the notion of Einstein manifold)
makes sense relative to any geometry. In fact, recall that tensor
geometries on $M$ are equivalent to reductions on the structural
group of $FM$, while connections compatible with tensors are equivalent
to Cartan connections for the corresponding group reductions. The
difference from the previous problematic approach is that, while $S_{EH}$
involves scalar curvature (which makes canonical sense only for particular
tensorial geometries), $S_{EHP}$ contains only the curvature form,
which can be defined for arbitrary Cartan connections. Not only this,
EHP theories can be defined in a purely algebraic sense: given an
$R$-algebra $A$ endowed with a vector decomposition $A\simeq A_{0}\oplus A_{1}$,
we can define a \emph{reductive $A$-connection} in a bundle $P\rightarrow M$
as a pair of $1$-forms $e:TP\rightarrow A_{0}$ and $\omega:TP\rightarrow A_{1}$.
The corresponding \emph{$A$-valued EHP theory} ($A$-EHP) is given
by the functional\footnote{In the way it is written, it only makes sense when the algebra $\Lambda(P;A)$
satisfies some associativity-like polynomial identity. We will see,
nonetheless, that it is possible to define $S_{EHP}$ for arbitrary
algebras.} 
\begin{equation}
S_{EHP}[e,\omega]=\int_{M}(\wedge_{*}^{n-2}e\wedge_{*}\Omega+\frac{\Lambda}{(n-1)!}\wedge_{*}^{n}e),\label{EHP}
\end{equation}
where $\wedge_{*}$ is the wedge product induced by the multiplication
$*:A\otimes A\rightarrow A$ and $\Omega=d\omega+\omega\wedge_{*}\omega$
is the \emph{field strength }(or \emph{curvature} of $\omega$). An
\emph{$A$-valued Einstein manifold} is then defined as a manifold
$M$ endowed with an $A$-connection $A=e+\omega$ fulfilling the
Euler-Lagrange equations of (\ref{EHP}). This definition can be enlarged
even more to include graded geometries, where $A$ is a graded algebra
and $P\rightarrow M$ is a graded bundle over a non-graded manifold.

Since EHP theories make sense in arbitrary geometries defined by abstract
algebras, we can ask if for each of those geometries the corresponding
theory is nontrivial in some sense, implying the existence of $A$-valued
Einstein manifolds. The minimal requirement is, of course, that the
functional must be non-null (recall that the symplectic version of
$S_{EH}$ vanishes identically, so that it is expected that a similar
situation will happen in EHP). In this article we give general obstruction
results for the minimal nontriviality of $A$-EHP (and, therefore,
for the existence of $A$-valued Einstein manifolds) in terms of properties
of the underlying algebra. We will prove two main theorems, which
are different in nature. The first of them applies to concrete and
abstract EHP functionals and gives obstructions on the dimension of
the base manifold, while the second applies only to concrete geometries,
but gives obstructions on the dimension and on the topology of the
base manifold. \vspace{0.2cm}

\noindent \textbf{Theorem A.} \emph{Let $P\rightarrow M$ be a graded
bundle over a $n$-dimensional manifold and $A\simeq\oplus_{m}A^{m}$
a graded $R$-algebra endowed with a $R$-module decomposition $A_{0}\oplus A_{1}$
such that each $A_{0}\cap A^m$ is a weak $(k_m,s_m)$-solv submodule
of $A_{0}$. Let $(k,s)$ be $\min_{m}(k_{m},s_{m})$. If $n\geq k+s+1$,
then the $A$-EHP action functional is homogeneous (i.e, $\Lambda=0$).
If $n\geq k+s+3$, then it is trivial.}\vspace{0.2cm}

\noindent \textbf{Theorem B. }\emph{Let $M$ be an $n$-dimensional
Berger manifold endowed with an $H$-structure. If there are natural
numbers $k_{1}+k_{r}=n$ such that $\mathfrak{h}\subset\mathfrak{so}(k_{1})\oplus...\oplus\mathfrak{so}(k_{r})$
properly, then the torsionless EHP functional is nontrivial only if
$n=2,4$ and $M$ is K\"{a}hler. In particular, if $M$ is compact
and $H^{2}(M;\mathbb{C})=0$, then it must be a K3-surface}.\vspace{0.2cm}

This paper is organized as follows: in Section \ref{sec_th_A} we
study the solvability conditions appearing in Theorem A and this theorem
is stated in a more general context and proved. In Section \ref{sec_th_B}
we recall some facts concerning algebra extensions and Theorem B is
stated also in a more general context and proved. In Section \ref{examples}
many examples are given, where by an example we mean a specific algebraic
context fulfilling the solvability hypothesis of Theorem A, so that
the corresponding EHP theory cannot be realized in the underlying
geometry. 

\section{Theorem A \label{sec_th_A}}

Let $(A,*)$ be an $R$-algebra, where $R$ is a commutative ring.
Given a smooth manifold $P$, let $\Lambda(P;A)$ denote the $\mathbb{N}$-graded
algebra of $A$-valued forms on $P$ with the wedge product $\wedge_{*}$
induced by $*$. An \emph{$A$-valued reductive connection} (or simply
$A$-connection) on a smooth bundle $P\rightarrow M$, relative to
a $R$-module decomposition $A\simeq A_{0}\oplus A_{1}$, is a pair
$\nabla=e+\omega$ of $1$-forms $e:TP\rightarrow A_{0}$ and $\omega:TP\rightarrow A_{1}$.
Therefore, the space of $A$-connections is given by 
\[
\mathrm{Conn}_{A}(P)=\Lambda^{1}(P;A_{0})\oplus\Lambda^{1}(P;A_{1}).
\]
The \emph{pure field strength }(or\emph{ curvature}) and the \emph{torsion}
of an $A$-connection $\nabla$ are respectively defined as the $A$-valued
$2$-forms $\Omega=d\omega+\omega\wedge_{*}\omega$ and $\Theta=de+\omega\wedge_{*}e$.
If $(A,*)$ is associative, then given $\Lambda\in\mathbb{R}$ there
is no ambiguity in considering the \emph{Hilbert-Palatini form} 
\[
\alpha_{n,\Lambda}=\wedge_{*}^{n-2}e\wedge_{*}\Omega+\frac{\Lambda}{(n-1)!}\wedge_{*}^{n}.
\]
In this case, fixed any FABS $\jmath$ and any trace transformation
$\mathrm{tr}$ (see Appendix, p. \pageref{appendix}) we have a natural
map $\mathrm{tr}\circ\jmath:\Lambda(P;A)\rightarrow\Lambda(M;\mathbb{R})$
preserving the $\mathbb{N}$-grading, allowing us to define the \emph{$A$-valued
EHP action functional} in $P\rightarrow M$ as the functional (if
$\Lambda=0$ we say that this is the homogeneous $A$-valued EHP functional)
\[
S_{n,\Lambda}:\mathrm{Conn}_{A}(P)\rightarrow\mathbb{R}\quad\text{such that}\quad S_{n,\Lambda}[e,\omega]=\int_{M}\mathrm{tr}(\jmath(\alpha_{n,\Lambda})).
\]

But, if $(A,*)$ does not satisfy associative-like polynomial identities
(PIs), in order to define Hilbert-Palatini forms and, consequently,
$A$-EHP theories, we need to work with $A$ up to those associative-like
PIs. More precisely, let us consider the associativity-like polynomials
\[
(x_{1}\cdot...\cdot(x_{s-2}\cdot(x_{s-1}\cdot x_{s}))),\quad(((x_{1}\cdot x_{2})\cdot(x_{3}\cdot x_{4}))\cdot...\cdot x_{s}),\quad\text{etc}.,
\]
Due to the structure of the Hilbert-Palatini forms, we are interested
in the more specific polynomials 
\begin{equation}
(x\cdot...\cdot(x\cdot(x\cdot y))),\quad\text{etc.}\quad\text{and}\quad(x\cdot...\cdot(x\cdot(x\cdot x))),\quad\text{etc}.\label{p_s}
\end{equation}
When evaluated at $\Lambda(P;A)$, the set of these polynomials generates
an ideal $\mathfrak{I}(P,A)$ and we can take the quotient $\Lambda(P;A)/\mathfrak{I}(P,A)$,
whose quotient map we denote by $\pi$. Fixing a FABS invariant under
the ideal $\mathfrak{I}(P,A)$ (see Appendix, p. \pageref{appendix})
we can now consider $S_{n,\Lambda}$ for arbitrary algebras: 
\begin{equation}
S_{n,\Lambda}[e,\omega]=\int_{M}\mathrm{tr}(\jmath(\pi(\alpha_{n,\Lambda}))).\label{algebraic_EHP}
\end{equation}

Suppose now that (as an $R$-module) $A$ is also $\mathfrak{m}$-graded
for some abelian group $(\mathfrak{m},+)$, meaning that $A\simeq\oplus_{m}A^{m}$
with $m\in\mathfrak{m}$. Such a decomposition induces a decomposition
in each $A_{i}$, with $i=0,1$, by $A_{i}\simeq\oplus_{m}A_{i}^{m}$,
where $A_{i}^{m}=A_{i}\cap A^{m}$. Suppose that $P\rightarrow M$
is also $\mathfrak{m}$-graded, in that $TP\simeq\oplus_{m}E^{m}$
for vector bundles $E^{m}\rightarrow M$. In this case, the previous
notions can also be generalized to the graded context. Indeed, given
$l\in\mathfrak{m}$ we define an \emph{$A$-connection of degree $l$}
\emph{in $P$} as a pair $\nabla=e+\omega$ of $\mathfrak{m}$-graded
$A$-valued $1$-forms of degree $l$, i.e, as a sequence of $1$-forms
$e^{m}:E^{m}\rightarrow A_{0}^{m+l}$ and $\omega^{m}:E^{m}\rightarrow A_{1}^{m+l}$.
The space of those connections is then 
\[
\mathrm{Conn}_{A}^{l}(P)=\oplus_{m}(\Lambda^{1}(E^{m};A_{0}^{m+l})\oplus\Lambda^{1}(E^{m};A_{1}^{m+l})).
\]

The curvature and the torsion are defined analogously, being determined
by the curvature and the torsion of each connection  $\nabla^{m}=e^{m}+\omega^{m}$.
As discussed in the Appendix, the notions of FABS and trace transformation
can be internalized into the category of graded principal bundles
(in the sense above), allowing us to define (for each graded-FABS
and each graded-trace transformations) the\emph{ $\mathfrak{m}$-graded
$A$-EHP functional of degree $l$ in $P\rightarrow M$} by the same
expression as (\ref{algebraic_EHP}), but now with domain $\mathrm{Conn}_{A}^{l}(P)$.

We will restate Theorem A considering the following more general conditions
on the $\mathfrak{m}$-graded algebra $A$ endowed with the a $R$-module
decomposition $A\simeq A_{0}\oplus A_{1}$:
\begin{enumerate}
\item[(G1)] $A_{0}$ is indeed a subalgebra and each $A_{0}^{j}=A^{i}\cap A_{0}$
is a $(k_{j},s_{j})$-weak solvable submodule of $A_{0}$.
\item[(G2)]  each $A^{j}$ is a $(k_{j},s_{j})$-weak solvable submodule of $A$.
\end{enumerate}
Before restating and proving Theorem A, let us explain what we mean
by a $(k,s)$-weak solvable submodule of an $R$-algebra $A$. We
say that an algebra $(A,*)$ is \emph{$(k,s)$-nil}, with $k,s\geq0$
if (for any smooth manifold $P$) every $A$-valued $k$-form in $P$
has nilpotency degree $s$, i.e, if for every $\alpha\in\Lambda^{k}(P;A)$
we have $\wedge_{*}^{s}\alpha\ne0$, but $\wedge_{*}^{s+1}\alpha=0$.
The $(k,s)$\emph{-solv} algebras are those that can be decomposed
as finite sums of $(k,s)$-nil algebras. We say that a submodule $V\subset A$
is a \emph{$(k,s)$-nil submodule of }$A$ if every $V$-valued $k$-form
has nilpotency degree $s$ when regarded as a $A$-valued form. Similarly,
we say that $V$ is a \emph{$(k,s)$-solv} \emph{submodule }if it
decomposes as a sum of $(k,s)$-nil submodules. When $A=\oplus_{m}A^{m}$
is $\mathfrak{m}$-graded, there are other kind of submodules $V\subset A$
that can be introduced. For instance, we say that $V$ is \emph{graded
$(k,s)$-solv} if each $V_{m}=V\cap A_{m}$ is a $(k,s)$-solv submodule.
In any submodule $V$ of a $\mathfrak{m}$-graded algebra we get a
corresponding grading by $V\simeq\oplus_{m}V_{m}$. So, any $V$-valued
form $\alpha$ can be written as $\alpha=\sum_{m}\alpha^{m}$. We
say that $V$ is \emph{weak $(k,s)$-nilpotent} if for every $k$-form
$\alpha$, any polynomial $p_{s+1}(\alpha^{m_{1}},...,\alpha^{s+1})$
vanishes. Similarly, we say that $V$ is \emph{weak $(k,s)$-solvable
}if it decomposes as a sum of weak $(k,s)$-nilpotent\emph{ }submodules.

Examples will be given in Section \ref{examples}. Just to mention,
if an algebra $A$ is nilpotent (resp. solvable) of degree $s$, then
it is $(1,s)$-nil (resp. $(1,s)$-solv). Similarly, degree $s$ nilpotent
(resp. solvable) subalgebras of any algebra are $(k,s)$-nil (resp.
$(k,s)$-solv) submodules.

\section*{Restatement and Proof}

Now we can give a more precise statement for Theorem A (in the formulation
below it is more general than that present in the Introduction).
\begin{thm}
\label{theorem_G1_2}Let $M$ be an $n$-dimensional manifold, $P\rightarrow M$
be a bundle such that $TP$ is $\mathfrak{m}$-graded and $(A,*)$
be a $\mathfrak{m}$-graded algebra endowed with a $R$-module decomposition
$A_{0}\oplus A_{1}$. Given $l\geq0$, assume that $A[-l]$ satisfies
condition\emph{ (G1)} or \emph{(G2)}. Furthermore, let $(k,s)$ be
the minimum of $(k_{j},s_{j})$. If $n\geq k+s+1$, then for any compatible
graded FABS and any trace transformation, the corresponding graded
$A$-EHP of degree $l$ equals the homogeneous ones. If $n\geq k+s+3$,
then $A$-EHP is trivial. 
\end{thm}

We will divide the proof in two cases:
\begin{enumerate}
\item when only the algebra $A$ is graded;
\item when both $A$ and $P$ are graded and the theory may have arbitrary
degree.
\end{enumerate}
\begin{proof}[Proof of first case.]
The proofs considering (G1) or (G2) are very similar, so we will
only explain the (G1) case. Since everything in the EHP action is
linear, grading-preserving and invariant under the action of $\mathfrak{I}(P;A)$,
it is enough to prove that $\alpha_{n,\Lambda}=0$ for some representative
Einstein-Hilbert form in the quotient $\Lambda(P;A)/\mathfrak{I}(P;A)$.
We choose
\[
\alpha_{n,\Lambda}=\wedge_{*}^{n-2}e\wedge_{*}\Omega+\frac{\Lambda}{(n-1)!}\wedge_{*}^{n}e.
\]
Under the hypothesis, we can locally write $e=\sum_{m}e^{m}$ with
$e^{m}:TP\rightarrow A_{0}^{m}$. From condition (G1) each $A_{0}^{m}$
is a weak $(k_{m},s_{m})$-solvable subspace, so that it writes as
a sum $A_{0}^{m}=\oplus_{i}V_{i}^{m}$ of weak $(k_{m},s_{m})$-nilpotent
subspaces, which means that we can write $e^{m}=\sum_{i}e_{i}^{m}$
locally and, therefore, $e=\sum_{m}\sum_{i}e_{i}^{m}$. Consequently,
for every $l$ we have $\wedge_{*}^{l}e=p_{l}(e_{i_{1}}^{m_{1}},...,e_{i_{l}}^{m_{l}})$
for some polynomial of degree $l$. If we now consider the minimum
$(k,s)$ (over $m$) of $(k_{m},s_{m})$, the fact that each $A_{0}^{m}$
is weakly $(k_{m},s_{m})$-nilpotent then implies $\wedge_{*}^{k+s+1}e=0$.
Consequently, 
\begin{equation}
(\wedge_{*}^{k+s+1}e)\wedge_{*}\alpha=0\label{obstruction_1}
\end{equation}
for any $A$-valued form $\alpha$. In particular, for 
\[
\alpha=\frac{\Lambda}{(n-1)!}\wedge_{*}^{n-(k+s+1)}e
\]
(\ref{obstruction_1}) is precisely the inhomogeneous part of $\alpha_{\Lambda,n}$.
Therefore, for every $n\geq k+s+1$ we have $\alpha_{n,\Lambda}=\alpha_{n}$
and consequently $\pi(\alpha_{n})=\pi(\alpha_{n,\Lambda})$, implying
$S_{n}[e,\omega]=S_{n,\Lambda}[e,\omega]$ for any $A$-valued connection,
and thus $S_{n}=S_{n,\Lambda}$. On the other hand, for
\[
\alpha=(\wedge_{*}^{n-(k+s+1)+2}e)\wedge_{*}\Omega
\]
we see that (\ref{obstruction_1}) becomes exactly $\alpha_{n}$.
Therefore, when $n\geq k+s+3$ we have $\alpha_{n}=0$, implying (because
$\jmath$ is linear) $\jmath(\alpha_{n})=0$ and thus $S_{n}=0$.
But $k+s+3>k+s+1$, so that $S_{n,\Lambda}=0$ too, ending this first
case.

$\;$

\noindent \emph{Proof of second case}. Notice that to give a morphism
$f:A'\rightarrow A$ of degree $l$ between two graded algebras is
the same as giving a zero degree morphism $f:A'\rightarrow A[-l]$,
where $A[-l]$ is the graded algebra obtained by shifting $A$. Consequently,
the obstructions of a degree $l$ graded $A$-EHP are just the obstructions
of degree zero graded $A[-l]$-EHP, so that it is enough to analyze
theories of degree zero. If we assume that $TP$ is a graded bundle,
so that $TP\simeq\oplus_{m}E^{m}$, and if $e$ has degree zero, then
the only change in comparison to the previous ``partially graded''
context is that instead of decomposing $e$ as a sum $\sum_{m}e^{m}$,
we can now write it as a genuine direct sum $e=\oplus_{m}e^{m}$,
with $e^{m}:E^{m}\rightarrow A_{0}\cap A^{m}$, meaning that the same
argumentation of the case case applies here, ending the proof. 
\end{proof}

\section{Theorem B \label{sec_th_B}}

Until now we worked in full generality and we showed that action functionals
defined on spaces of algebra valued forms suffer minimal obstructions
(in terms of the properties of the algebra) which affect the possible
dimensions of the base manifold. Here we will show that if we restrict
ourselves to certain concrete situations the base manifold suffers
not only those minimal obstructions, but also strong topological obstructions.

Recall that given algebras $E$, $A$ and $H$, we say that $E$ is
an \emph{extension of $A$ by $H$} when they fit into a short exact
sequence (as $R$-modules):$$
\xymatrix{ 0 \ar[r] & H \ar[r]^{\imath} & E \ar[r]^{\pi} & A \ar[r] & 0}
$$ Furthermore, the extension is called \emph{separable} if $\pi$ has
a section, i.e, a $R$-linear map $s:A\rightarrow E$ such that $\pi\circ s=id_{A}$.
In this case we can use $s$ to pushforward the product of $A$ to
$E$, as below:\begin{equation}{
\label{pullback_product}
\xymatrix{ \ar@/_{0.5cm}/[rrr]_-{*'} E\otimes E \ar[r]^-{\pi\otimes \pi} & A\otimes A \ar[r]^-{*} & A \ar[r]^-s & E }}
\end{equation}

Concerning this new product, all that we need is the following result:
\begin{lem}
\label{pullback_lemma}The algebra $(A,*)$ is $(k,s)$-nil iff $(E,*')$
is $(k,s)$-nil.
\end{lem}

\begin{proof}
Relative to the new product $A$ is clearly a subalgebra of $E$.
Therefore, $A$ is $(k,s)$-nil only if $(E,*')$ is. In order to
get the reciprocal, recall that any algebra $A$ induces a corresponding
graded-algebra structure in $\Lambda(P;A)$. Due to the functoriality
of $\Lambda(P;-)$, we then get the commutative diagram below, where
the horizontal rows are just (\ref{pullback_product}) for the algebras
$(E,*)$ and $(A,*')$, composed with the diagonal map. The commutativity
of this diagram says just that $\wedge_{*'}^{2}\alpha=\wedge_{*}^{2}f(\alpha)$
for every $\alpha\in\Lambda(P;A)$. From the same construction we
get, for each given $s\geq2$, a commutative diagram that implies
$\wedge_{*'}^{s}\alpha=\wedge_{*}^{s}f(\alpha)$. Therefore, if $(A,*)$
is $(k,s)$-nil it immediately follows that $(E,*')$ is $(k,s)$-nil
too. $$
\xymatrix{\ar[d]<0.2cm>^{\Lambda \pi} \Lambda (P;E) \ar[r]^-{\Delta} & \ar[d]<0.2cm>^{(\Lambda \pi) \otimes (\Lambda \pi)} \Lambda (P;E)\otimes \Lambda (P;E) \ar[r] & \ar[d]<0.2cm>^{\Lambda (\pi\otimes \pi)} \Lambda (P;E\otimes E) \ar[r] & \ar[d]<0.2cm>^{\Lambda \pi} \Lambda (P;E) \\
\ar[u]<0.2cm>^{\Lambda s} \Lambda (P;A) \ar[r]^-{\Delta} & \Lambda (P;A)\otimes \Lambda (P;A) \ar[r] \ar[u]<0.2cm>^{(\Lambda s) \otimes (\Lambda s)} & \Lambda (P;A\otimes A) \ar[r] \ar[u]<0.2cm>^{\Lambda (s\otimes s)} & \ar[u]<0.2cm>^{\Lambda s} \Lambda (P;A)}
$$
\end{proof}
In the following we will be interested in:
\begin{enumerate}
\item \emph{$k$-algebras}. These are algebras $(E,*')$ arising as splitting
extensions of real matrix algebras $(A,*)$ which are associative
subalgebras of $\mathfrak{so}(k_{1})\oplus...\oplus\mathfrak{so}(k_{r})$
for some $k_{1},..,k_{r}$. We notice that $\mathfrak{so}(k)$ is
$(1,2)$-nil for any $k$. It then follows that $A$ (and therefore
$E$, due to lemma above) is $(1,2)$-nil, because $A\subset\mathfrak{so}(k_{1}+..+k_{r})$;
\item \emph{Berger $k$-manifolds} $N$. These have dimension $k$ and are
simply connected, locally irreducible and locally non-symmetric;
\item \emph{$k$-proper bundles} $P\rightarrow M$. These are such that
there exists an immersed Berger $k$-manifold $N\hookrightarrow M$
whose frame bundle is a subbundle of $P$, i.e, such that $FN\subset\imath^{*}P$.
The motivating examples are the frame bundle of a Berger $k$-manifold
and, more generally, the pullback of the frame bundle by the immersion
$f:N\rightarrow M$ of a Berger $k$-manifold.
\item \emph{torsion-free EHP theories}. By this we mean $A$-EHP funcional
(\ref{algebraic_EHP}) restricted to the subspace $\mathrm{Conn}_{A}^{\Theta=0}(P)\subset\mathrm{Conn}_{A}(P)$
of $A$-valued connections $\nabla=e+\omega$ such that $\Theta_{\omega}=0$.
\end{enumerate}

\section*{Restatement and Proof}

Let us now restate and proof Theorem B and analyze its fundamental
consequences.
\begin{thm}
\label{holonomy_EHP}Let $P\rightarrow M$ be a $k$-proper $H$-bundle
over a $n$-manifold $M$. If $\mathfrak{h}$ is a $k$-algebra, then
$k_{1}=k$ and $k_{i>1}=0$. Furthermore, for any FABS, the torsionless
EHP theory with values into $\mathbb{R}^{k}\rtimes\mathfrak{h}$ is
nontrivial only if one of the following conditions is satisfied 
\begin{enumerate}
\item[\textit{\emph{(B1)}}] $k=2,4$ and $M$ contains a K\"{a}hler Berger $k$-manifold; 
\item[\textit{\emph{(B2)}}] $k=4$ and $M$ contains a quaternionic-K\"{a}hler Berger $k$-manifold. 
\end{enumerate}
\end{thm}

\begin{proof}
Since $\mathfrak{h}$ is a $k$-algebra, it is $(1,2)$-nil and, therefore,
the splitting extension $\mathbb{R}^{k}\rtimes\mathfrak{h}$ is $(1,2)$-nil
too. Consequently, by Theorem \ref{obstruction_1} the actual (and,
in particular, the torsionless) EHP is trivial if $n\geq6$, so that
we may assume $n<6$. Since $P\rightarrow M$ is $k$-proper, $M$
contains at least one immersed Berger $k$-manifold $\imath:N\hookrightarrow M$
such that $FN\subset\imath^{*}P$. Let $\kappa$ denote the inclusion
of $FN$ into $\imath^{*}P$. On the other hand, we also have an immersion
$\imath^{*}P\hookrightarrow P$, which we denote by $\imath$ too.
Lie algebra-valued forms can be pulled-back and the pullback preserves
horizontability and equivariance. So, for any $\nabla=e+\omega$ the
corresponding $1$-form $(\imath\circ\kappa)^{*}\omega\equiv\omega\vert_{N}$
is an $H$-connection in $FN$ whose holonomy is contained in $H$.
Since $H$ is the Lie integration of a $k$-algebra we have 
\[
H\subset SO(k_{1})\times...\times SO(k_{r})\subset SO(k).
\]
In particular, the holonomy of $\omega\vert_{N}$ is contained in
$SO(k)$, implying that $\omega\vert_{N}$ is compatible with some
Riemannian metric $g$ in $N$. But, we are working with torsion-free
connections, so that $\omega\vert_{N}$ is actually the Levi-Civita
connection of $g$ and, because $N$ is irreducible, de Rham decomposition
theorem implies that there exists $i\in1,...,r$ such that $k_{i}=k$
and $k_{j}=0$ for $j\neq i$. Without loss of generality we can take
$i=1$. Because $N$ is a Berger manifold, Berger's theorem \citep{berger_original}
applies, implying that the holonomy of $\omega\vert_{N}$ is classified,
giving conditions (B1) and (B2). 
\end{proof}
\begin{cor}
\label{corollary_holonomy}Let $M$ be a Berger $n$-manifold with
an $H$-structure, where $H$ is such that $\mathfrak{h}$ is a $n$-algebra.
In this case, for any FABS, the EHP theory with values into $\mathbb{R}^{n}\rtimes H$
is nontrivial only if $M$ has dimension $n=2,4$ and admits a K\"{a}hler
structure. 
\end{cor}

\begin{proof}
The result follows from the last theorem by considering the bundle
$P\rightarrow M$ as the frame bundle $FM\rightarrow M$ and from
the fact that every orientable four-dimensional smooth manifold admits
a quaternionic-K\"{a}hler structure \citep{einstein_manifolds_BESSE,quaternionic_Kahler_SALAMON}. 
\end{proof}
This corollary shows how topologically restrictive it is to internalize
\emph{torsionless} EHP in geometries other than Lorentzian. Indeed,
if the manifold $M$ is compact and 2-dimensional, then it must be
$\mathbb{S}^{2}$. On the other hand, in dimension $n=4$ compact
K\"{a}hler structures exist iff the Betti numbers $b_{1}(M)$ and
$b_{3}(M)$ are zero, so that $\chi(M)=b_{2}(M)+2$. As a consequence,
if we add the (mild) condition $H^{2}(M;\mathbb{R})\simeq0$ on the
hypothesis of Corollary \ref{corollary_holonomy} we conclude that
$M$ must be a K3-surface!

\section{Examples \label{examples}}

In this section we will give realizations of the obstruction theorems
studied previously. Due to the closeness with Lorentzian EHP (which
is where EHP theories are usually formulated), our focus is on concrete
situations, meaning that we will be focused in $A$-EHP functionals
for $A$ a $k$-algebra or an splitting extension of a $k$-algebra.
Even so, some abstract examples will also be given.

\section*{Concrete Examples\label{linear_examples}}

Let us start by considering $A$-EHP for $A$ an $\mathbb{R}$-algebra
endowed with a distinguished subalgebra $\overline{A}\subset A$.
We then have two vector space decompositions $A\simeq A_{0}\oplus A_{1}$:
one for which $A_{0}=\overline{A}$ and $A_{1}=A/\overline{A}$ and
another one for the opposite. From Theorem \ref{theorem_G1_2} and
the fact that $k$-algebras are $(1,2)$-nil we have the following
conclusions for $n\geq6$:
\begin{enumerate}
\item[(E1)]  if $\overline{A}$ is an ideal and $A/\overline{A}$ is a $k$-algebra,
then the EHP funcional relative to the opposite decomposition is trivial;
\item[(E2)] if $\overline{A}$ is a $k$-subalgebra, then the EHP functional
relative to the first decomposition is trivial; 
\item[(E3)] if $A$ is $k$-algebra, then EHP is trivial in both decompositions. 
\end{enumerate}
Conditions (E2) and (E3) are immediately satisfied for $k$-groups,
i.e, Lie subgroups of $SO(k_{1})\times...\times SO(k_{r})$, or, equivalently,
Lie subalgebras of $\mathfrak{so}(k_{1})\oplus...\oplus\mathfrak{so}(k_{r})$.
Indeed, if $G_{0}$ is a such subgroup, then the EHP theory for any
associative algebra algebra $A$ extending $\mathfrak{g}_{0}$ obviously
satisfies (E2), so that if $n\geq6$ the dual theories are trivial.
On the other hand, for any subalgebra $A_{0}\subset\mathfrak{g}_{0}$
the EHP theory with values into $\mathfrak{h}$ clearly satisfies
(E3) so that in dimension $n\geq6$ both the dual and the actual EHP
theory are trivial. Some obvious examples of subgroups of $k$-groups
$G_{0}\subset SO(k)$ are given in the table below. Except for $Spin(4)\hookrightarrow O(8)$,
which arises from the exceptional isomorphism $Spin(4)\simeq SU(2)\times SU(2)$,
all the other inclusions appear in Berger's classification \citep{berger_original}. 

Let us focus on (E2). We can think of each element of Table \ref{first_examples}
as included in some $GL(k;\mathbb{R})$, i.e, as a $G$-structure
on a manifold and then as a geometry. We could get more examples by
taking finite products of arbitrary elements in the table. In terms
of geometry, this can be interpreted as follows. Recall that a distribution
of dimension $k$ on an $n$-manifold can be regarded as a $G$-structure
for $G=GL(k)\times GL(n-k)$. Therefore, we can think of a product
$O(k)\times O(n-k)$ as a distribution of Riemannian leaves on a Riemannian
manifold, $U(k)\times U(n-k)$ as a hermitan distribution, and so
on.
\begin{table}[H]
\begin{centering}
\begin{tabular}{|c|c|c|c|c|c|}
\hline 
$(n,k)$  & $(1,k)$  & $(2,k)$  & $(4,k)$  & $(1,7)$  & $(1,8)$\tabularnewline
\hline 
$H\subset O(n\cdot k)$  & $SO(k)$  & $U(k)$  & $Sp(k)$  & $G_{2}$  & $Spin(4)$\tabularnewline
\hline 
\multicolumn{1}{c|}{} & $U(k)$  & $SU(k)$  & $Sp(k)\cdot Sp(1)$  &  & $Spin(7)$\tabularnewline
\cline{2-6} 
\end{tabular}
\par\end{centering}
\caption{\label{first_examples}First examples of concrete geometric obstructions}
\end{table}

Other special cases where condition (E2) applies are in the table
below. In the first line, $O(k,k)$ is the so-called \emph{Narain
group} \citep{narain_group}, i.e, the orthogonal group of a metric
with signature $(k,k)$, whose maximal compact subgroup is $O(k)\times O(k)$.
Second line follows from an inclusion similar to $U(k)\hookrightarrow O(2k)$,
first studied by Hitchin and Gualtieri \citep{generalized_complex_PRIMEIRO,generalized _complex_thesis,generalized_complex_2},
while the remaining lines are particular cases of the previous ones.
The underlying flavors of geometry arose from the study of Type II
gravity and Type II string theory \citep{generalized_geometry_TYPE_II,generalized_geometry_TYPE_1}.
That condition (E2) applies for the second and third lines follows
from the fact that by complexifying $U(k,k)\hookrightarrow O(2k,2k)$
we obtain $U(k,k)\hookrightarrow O(4k;\mathbb{C})$. 
\begin{table}[H]
\begin{centering}
\begin{tabular}{|c|c|c|}
\hline 
$G$  & $G_{0}$  & $\text{geometry}$\tabularnewline
\hline 
\hline 
$O(k,k)$  & $O(k)\times O(k)$  & Type II\tabularnewline
\hline 
$O(2k,2k)$  & $U(k,k)$  & Generalized Complex\tabularnewline
\hline 
$O(2k,2k)$  & $SU(k,k)$  & Generalized Calabi-Yau\tabularnewline
\hline 
$O(2k,2k)$  & $U(k)\times U(k)$  & Generalized K\"{a}hler\tabularnewline
\hline 
$O(2k,2k)$  & $SU(k)\times SU(k)$  & Generalized Calabi\tabularnewline
\hline 
\end{tabular}
\par\end{centering}
\caption{\label{more_examples}More examples of concrete geometric obstructions}
\end{table}

About these two tables, some remarks: 
\begin{enumerate}
\item Table \ref{first_examples} contains any classical flavors of geometry,
except symplectic and complex. The reason is that $\operatorname{Sp}(k;\mathbb{R})$
and $GL(k;\mathbb{C})$ are not contained in some $O(r)$. But this
does not mean that condition (E2) is not satisfied by those geometries
Indeed, generalized complex geometry contains both of them \citep{generalized _complex_thesis},
so that they actually fulfill (E2). 
\item We could create a third table with ``exotic $k$-groups'', meaning
that a priori they are not related to any classical geometry, so that
they describe some kind of ``exotic geometry''. For instance, in
\citep{more_subgroups_O(n)} all Lie subgroups $G_{0}\subset O(k)$
satisfying 
\[
\frac{(k-3)(k-4)}{2}+6<\dim H<\frac{(k-1)(k-2)}{2}
\]
were classified and in arbitrary dimension $k$ there are fifteen
families of them. Other exotic (rather canonical, in some sense) subgroups
that we could add are maximal tori. Indeed, both $O(2k)$ and $U(k)$
have maximal tori, say denoted by $T_{O}$ and $T_{U}$, so that the
reductions $T_{O}\hookrightarrow O(2k)$ and $T_{U}\hookrightarrow U(k)$
will satisfy condition (E2). More examples of exotic subgroups to
be added are the \emph{point groups}, i.e, $H\subset\operatorname{Iso}(\mathbb{R}^{k})$
fixing at least one point. Without loss of generality we can assume
that this point is the origin, so that $H\subset O(k)$. Here we have
the symmetric group of any spherically symmetric object in $\mathbb{R}^{k}$,
such as regular polyhedrons and graphs embedded on $\mathbb{S}^{k-1}$. 
\item If we are interested only in condition (E2), then the tables above
can be enlarged by including embeddings of $O(k_{1})\times...\times O(k_{r})$
into some other larger group $\tilde{G}$. In this condition it only
matters that $H$ is a $k$-group. Particularly, $O(k)$ admits some
exceptional embeddings (which arise from the classification of simple
Lie algebras), as in the table below. 
\begin{table}[H]
\begin{centering}
\begin{tabular}{|c|c|c|c|c|c|}
\hline 
$k$  & 3  & 9  & 10  & 12  & 16\tabularnewline
\hline 
$O(k)\hookrightarrow$  & $G_{2}$  & $F_{4}$  & $E_{6}$  & $E_{7}$  & $E_{8}$\tabularnewline
\hline 
\end{tabular}
\par\end{centering}
\caption{\label{exceptional_embed} Exceptional embbedings of orthogonal groups.}
\end{table}
\item Despite (E2) making sense sense for arbitrary real algebras $A$ with
distinguished $k$-subalgebra $A_{0}\subset A$, the examples above
were all for matrix algebras underlying Lie algebras of matrix Lie
groups. But, as pointed out previously, in order to get (E2) we only
need that $A_{0}$ be a $k$-algebras, meaning that \textbf{any} extension
(not necessarily a matrix extension) will satisfy (E2). In particular,
for every $G_{0}$ in the above tables $\mathbb{R}^{k}\rtimes\mathfrak{g}_{0}$
satisfies this condition. Recall that EHP theories with values in
those algebras are immediate generalizations of Lorentzian EHP theory.
Therefore, what we showed is that the EHP functional cannot be realized
(in dimension $n\geq6$) in essentially all flavors of geometry other
then semi-Riemannian with specific signature. 
\item Condition (E2) is on the Lie algebra level. A Lie group and its universal
covering have the same Lie algebra. Therefore, we can double the size
of our current tables by adding the universal cover (when it exists)
of each group. For instance, in Table \ref{first_examples} the spin
groups $\operatorname{Spin}(4)$ and $\operatorname{Spin}(7)$ were
added due to exceptional isomorphisms on the level of Lie groups.
Now, noticing that for $k>2$ (in particular for $k=4,7$) $\operatorname{Spin}(k)$
is the universal cover of $SO(k)$, we can automatically add all spin
groups to our table. We can also add the universal covering of the
symplectic group $Sp(k;\mathbb{R})$, the metapletic group.
\item Due to the above discussion, we can use Tables \ref{first_examples}
and \ref{more_examples}, as well as the ``table of exotic $k$-groups'',
as a source of examples for condition (E3). Indeed, we can just consider
the algebra $A$ in condition (E3) as the algebra of some $H$ in
the mentioned tables and take a fixed subalgebra $A_{0}\subset\mathfrak{g}$.
This produces a long list of examples, because $G/G_{0}$ is a priori
an arbitrary homogeneous space subjected only to the condition of
$G$ being a $k$-group (here $G_{0}$ is the integrated Lie group
of $A_{0}$). In geometric terms, if a geometry fulfills condition
(E2), then any induced ``homogeneous geometry'' fulfills (E3). On
the other hand, differently from what happens with condition (E2),
Table \ref{exceptional_embed} \emph{cannot} be used to get more examples
of condition (E3). Indeed, if $G_{0}\hookrightarrow G$ is a reduction
fulfilling (E3) and $G\hookrightarrow\tilde{G}$ is an embedding,
then $G_{0}\hookrightarrow\tilde{G}$ fulfills (E3) iff it is fulfilled
by $G\hookrightarrow\tilde{G}$ (see diagram below).$$
\xymatrix{H \ar@/_{0.6cm}/[rrr]_{} \ar@{^(->}[r] & G \ar@{^(->}[r] & O(k_{1})\times...\times O(k_{r}) \ar@{^(->}[r] & \tilde{G}  }
$$ 
\end{enumerate}
Let us now analyze condition (E1). First of all we notice that it
is very restrictive, because we need to work with subalgebras $A_{0}\subset A$
of $k$-algebras which are ideals. In terms of Lie groups this means
that we need to restrict to normal subgroups $G_{0}\subset G$ such
that the quotient $G/G_{0}$ is a $k$-group. For instance, in the
typical situations above, $G_{0}$ is not normal. Even so, there are
two dual conditions under which $G_{0}\subset G$ fulfills (E1): when
$G$ is a $k$-group with $G/G_{0}\subset G$ and, dually, when $G_{0}$
is $k$-group with $G/G_{0}\subset G_{0}$. However, in both situations
we are also in (E3), so that there is nothing much new here.

\section*{Cayley-Dickson Examples}

Here we will show that the list of concrete examples can be extended
indefinitely, allowing us to conclude that EHP theories are trivial
in an infinite number of geometries. This will be done by making use
of the Cayley-Dickson construction. Let $(A,*)$ be an arbitrary $\mathbb{R}$-algebra
with involution $\overline{(-)}:A\rightarrow A$. For each fixed $l>0$
we get a sesquilinear map $s:A^{l}\times A^{l}\rightarrow A$ such
that $s(x,y)=\sum_{i}\overline{x}_{i}*y_{i}$. Consider the subspace
of all $l\times l$ matrices $M\in\operatorname{Mat}_{l\times l}(A)$
with coefficients in $A$ which preserve $s$, in the sense that $s(Mx,My)=s(x,y).$
We say that these are the \emph{unitary matrices in $A$, }and we
denote this set by $U(k;A)$. If the involution is trivial (i.e, the
identity map), then call them the \emph{orthogonal matrices} \emph{in
$A$}, writing $O(k;A)$ to denote the corresponding space. 

Recall that Cayley-Dickson takes an involutive algebra $A$ and gives
another involutive algebra $\operatorname{CD}(A)$ with weakened PIs
\citep{Cayley_Dickson_2}. As an $\mathbb{R}$-module, the newer algebra
is given by a sum $A\oplus A$ of real and imaginary parts. We have
an inclusion $A\hookrightarrow\operatorname{CD}(A)$, obtained by
regarding $A$ as the real part, which induces an inclusion into the
corresponding unitary groups $U(k;A)\hookrightarrow U(k;\operatorname{CD}(A))$.
This inclusion can be extended in the following way 
\[
U(k;A)\hookrightarrow U(k;\operatorname{CD}(A))\hookrightarrow U(2k;A),
\]
defined by setting the real part and the imaginary part as diagonal
block matrices. Iterating we see that for every $k$ and every $l$
there exists an inclusion 
\begin{equation}
U(k;\operatorname{CD}^{l}(A))\hookrightarrow U(2^{l}k;A).\label{inclusion_unitary}
\end{equation}

We assert that, at least when the starting algebra $A$ is finite-dimensional,
each $U(k;\operatorname{CD}^{l}(A))$ is a Lie group. Taking $A=\mathbb{R}$
in the inclusion above we will then get a corresponding sequence of
matrix Lie algebras (and, therefore, a sequence of underlying associative
algebras). Consequently, each of such algebras will satisfy condition
(E2), meaning that, for every $k>0$ and $l\ge0$, the EHP theories
with values into extensions of $\mathfrak{u}(k;\operatorname{CD}^{l}(A))$
are trivial.
\[
\xymatrix{U(k;\operatorname{CD}^{l}(\mathbb{R}))\ar[r] & \cdots\ar[r] & U(2^{l-2}k;\mathbb{H})\ar[r] & U(2^{l-1}k;\mathbb{C})\ar[r] & O(2^{l}k)}
.
\]

The idea for proving the assertion is to reproduce the proof that
the standard unitary/orthogonal groups $U(k)$ and $O(k)$ are Lie
groups. The starting point is to notice that the involution of $A$
induces an involution $(-)^{\dagger}$ in $\operatorname{Mat}_{k\times k}(A)$,
defined by composition with the transposition map. The next step is
to see that $k\times k$ matrices preserve the involution $s$ iff
they are invertible with $M^{-1}=M^{\dagger}$, i.e, iff $M\cdot M^{\dagger}=1_{k}=M^{\dagger}\cdot M$.
Since $\overline{(-)}:A\rightarrow A$ is an algebra morphism, we
have the usual property $(M\cdot N)^{\dagger}=M^{\dagger}\cdot N^{\dagger}$,
allowing us to characterize the unitary matrices as those satisfying
$M\cdot M^{\dagger}=1_{k}$. Therefore, defining the map 
\[
f:\operatorname{Mat}_{k\times k}(A)\rightarrow\operatorname{Mat}_{k\times k}(A)\quad\text{by}\quad f(M)=M\cdot M^{\dagger},
\]
in order to proof that $U(k;A)$ is a Lie group it is enough to verify
that the above map is, in some sense, a submersion (which will give
the smooth structure) and that the multiplication and inversion maps
are smooth. When $A$ is finite-dimensional this is immediate (although
this requeriment is not necessary).

Since classical geometries (meaning geometries in the sense of Weil)
are described by inclusions of Lie groups, this leads us to think
of (\ref{inclusion_unitary}) as modeling some flavor of geometry,
say the \emph{Cayley-Dickson geometry in $A$, of order }$l$\emph{.
}With this nomenclature we have concluded that EHP theories are trivial
in real Cayley-Dickson geometries of arbitrary order, generalizing
the examples in Table \ref{first_examples}. In the following we will
show that Table \ref{more_examples} can be generalized as well. In
fact we will show that there is a notion of \emph{generalized Cayley-Dickson
geometry in $A$ }which, for $A=\mathbb{R}$, contains Hitchin's generalized
complex geometry and in which EHP theories are trivial. This follows
directly from the fact that for every $A$ we have the inclusion $U(k,k;\operatorname{CD}(A))\subset O(2k;2k,A)$,
so that by iteration we get 
\begin{equation}
U(k,k;\operatorname{CD}^{l}(A))\subset O(2^{l}k;2^{l}k,A).\label{generalized_CD}
\end{equation}

For $A=\mathbb{R}$ and $l=1$ this model generalized complex geometry,
leading us to say that the inclusion above models \emph{generalized
Cayley-Dickson geometry in $A$ of degree $l$}. For instance, if
we take $A=\mathbb{R}$ and $l=2$ this becomes generalized quaternionic
geometry, which is a poorly studied theory, beginning with the works
\citep{generalized_quaternionic_INITIAL,generalized_quaternionic_1}
. For $l=3,4,...$ it should be generalized octonionic geometry, generalized
sedenionic geometry, and so on. The authors are unaware of the existence
of substantial works on these theories.

When tensoring inclusion (\ref{generalized_CD}) with $\operatorname{CD}^{l}(A)$
we get 
\[
U(k,k;\operatorname{CD}^{l}(A))\subset O(2^{l}k;2^{l}k,\operatorname{CD}^{l}(A))\simeq O(2^{l}k+2^{l}k;\operatorname{CD}^{l}(A)),
\]
so that condition (E2) above implies that, independently of the present
development of generalized Cayley-Dickson geometries, EHP theory is
trivial in each of them.

\section*{Abstract Examples}

Here we give abstract examples of algebras in which EHP cannot be
realized. By abstract we mean that they are not necessarily $k$-algebras.
We will restrict to the cases in which only $A$ is $\mathfrak{m}$-graded.
From Theorem \ref{theorem_G1_2} we have that:
\begin{enumerate}
\item[(F1)] \emph{\label{condition_F1}}when $A$ admits a vector space decomposition
$A\simeq A_{0}\oplus A_{1}$, where $A_{0}$ is a subalgebra such
that each $A_{0}^{m}=A_{0}\cap A^{m}$ is a weak $(k_{m},s_{m})$-solvable
subspace, then $A$-EHP is trivial in dimension $n\geq k+s+3$, where
$(k,s)$ is the minimum over $(k_{m},s_{m})$. 
\end{enumerate}
The most basic examples are those for $\mathfrak{m}=0$ and $A_{1}=0$,
i.e, the non-graded setting with $A$ itself $(k,s)$-solv. For instance,
any Lie algebra is $(k,1)$-nil for any even $k$. As a consequence,\emph{
in dimension $n\geq4$, any Lie algebra valued EHP functional is trivial.
}In particular, EHP with values in the Poincaré group $\mathfrak{iso}(n-1,1)$
is trivial. We emphasize that this \textbf{does not} mean that the
standard Lorentzian EHP theory is trivial (which would be absurd).
Indeed, while Lorentzian EHP theory and $\mathfrak{iso}(n-1,1)$-valued
EHP theory take values in the same \textbf{vector space}, and their
action functionals have the same shape, the \textbf{algebras} (and,
therefore, their properties) used to define the corresponding wedge
product are \emph{totally different}. 

If we now allow $\mathfrak{m}$ to be nontrivial, but with $A_{1}=0$,
then condition (F1) is satisfied if each $A^{m}$ is a weak $(k_{m},s_{m})$-solvable
subespace. In particular, it remains satisfied if $A$ is itself $(k,s)$-solv.
It happens that not only Lie algebras are solv $(2,1)$-nil, but also
graded Lie algebras. Consequently, EHP theories are also trivial in
the domain of graded Lie algebras. One can generalize even more thinking
of EHP theories with values in Lie superalgebras and in graded Lie
superalgebras. Indeed, recall that a \emph{$\mathfrak{m}$-graded
Lie superalgebra $\mathfrak{g}$} is just a $\mathfrak{m}$-graded
Lie algebra whose underlying PIs (i.e, skew-commutativity and Jacobi
identity) hold in the graded sense. Particularly, this means that
the $\mathbb{Z}_{2}$-grading writes $\mathfrak{g}\simeq\mathfrak{g}^{0}\oplus\mathfrak{g}^{1}$,
with $\mathfrak{g}^{0}$ a $\mathfrak{m}$-graded Lie algebra and,
therefore, $(2,1)$-nil. It then follows that $\mathfrak{g}$ satisfies
condition (F1). Summarizing, in dimension $n\geq6$,\emph{ the} \emph{EHP
functional cannot be realized in any Lie algebraic context}.

Graded Lie superalgebras are the first examples of algebras satisfying
(F1) with $A_{1}\neq0$, but they are far from being the only one.
Indeed, when we look at a decomposition $A\simeq A_{0}\oplus A_{1}$,
where $A_{0}$ is a subalgebra, it is inevitable to think of $A$
as an extension of $A_{1}$ by $A_{0}$, meaning that we have an exact
sequence as shown below. If $A_{0}$ is $(k,s)$-solv, then (F1) holds.
This can be interpreted as follows: \emph{ }suppose that we encountered
an algebra $A_{1}$ such that EHP is not trivial there. So, EHP will
be trivial in any (splitting) extension of $A_{1}$ by a $(k,s)$-solv
algebra.
\[
\xymatrix{0\ar[r] & A_{0}\ar[r] & A\ar[r] & A_{1}\ar[r] & 0}
\]

In particular, because $(\mathbb{R}^{k},+)$ is abelian and, therefore,
$(1,1)$-nil, any algebra extension by $\mathbb{R}^{k}$ will produce
a context in which EHP is trivial. Recall the previous concrete examples,
in which we considered splitting extensions of matrix algebras by
$\mathbb{R}^{k}$. This would seem to suggest that every concrete
EHP theory is trivial, which is not true. Once again the problem is
in the wedge products: we have the same vector space structures, but
in the concrete contexts and here the way of multiplying forms with
values in those spaces are different. The main difference is that
the product used there does not take the abelian structure of $\mathbb{R}^{k}$
into account. 

Another example about extensions is as follows: the paragraph above
shows that abelian extensions of nontrivial EHP are trivial, but what
about super extensions? They remain trivial. Indeed, being a Lie superalgebra,
the translational superalgebra $\mathbb{R}^{k\vert l}$ (the cartesian
superspace, regarded as a Lie superalgebra) is $(2,1)$-nil, so that
any algebra extension by it is trivial. So, if we think on supergravity
as a realization of gravity not in the domain of Cartan connections
but in the domain of super Cartan connections on the super Poincaré
Lie algebra $\mathfrak{siso}(n-1,1)$ we find that EHP functional
is not a good model, because it will be trivial if $n\geq6$.

Finally, there are also purely abstract examples. Just to mention,
in \citep{exotic_nil_1} the authors build ``mathematically exotic''
examples of nilpotent algebras, which fulfill (F1). In \citep{maximal_tori}
necessary and sufficient conditions were given under which arbitrary
orthogonal groups $O(q,\mathbb{K})$, where $q:V\rightarrow\mathbb{K}$
is a positive-definite quadratic form on a finite-dimension $\mathbb{K}$-space
$q$ and $\operatorname{ch}(\mathbb{K})\neq2$, admit an embedded
maximal torus $\mathbf{T}(q;\mathbb{K})$. Independently of the quadratic
space $(V,q)$, the corresponding orthogonal group is a Lie group
and, as in the real case, the algebra $\mathfrak{o}(q;\mathbb{K})$
is $(2,1)$-nil. Therefore, under the conditions of \citep{maximal_tori},
the EHP functional will be trivial in those toroidal geometries. 

\section{Appendix \label{appendix}}

Here we will introduce a general tool that allows us to pushforward
algebra-valued forms on the total space $P$ of a bundle $P\rightarrow M$
to the base manifold $M$. The first step is to build some process
allowing us to replace $A$-valued forms in $P$ by forms in $M$
with coefficients in some other bundle, say $E_{A}$. It is more convenient
to think of this in categorical terms. Let $\mathbf{Alg}_{R}$ be
the category of $R$-algebras, $\mathbb{Z}\mathbf{Alg}_{R}$ be the
category of $\mathbb{Z}$-graded $R$-algebras and, given a manifold
$M$, let $\mathbf{Bun}_{M}$ and $\mathbf{Alg_{\text{\ensuremath{R}}}Bun}_{M}$
denote the categories of bundles and of $R$-algebra bundles over
$M$, respectively. A \emph{functorial algebra bundle system }(FABS)
consists of 
\begin{enumerate}
\item a subcategory $\mathbf{C}\subset\mathbf{Bun}_{M}\times\mathbb{Z}\mathbf{Alg}_{R}$; 
\item a functor $E_{-}$ assigning to any algebra $A\in\mathbf{C}$ a corresponding
algebra bundle $E_{A}\rightarrow M$ whose typical fiber is $A$; 
\item a functor $S(-;-)$ that associates an algebra to each pair $(P,A)\in\mathbf{C}$; 
\item natural transformations $\imath:S(-;-)\Rightarrow\Lambda(-;-)$ and
$\xi:S(-;-)\Rightarrow\Lambda(M;-)$ such that $\imath$ is objectwise
injective.
\end{enumerate}
Now that we know how to replace $A$-valued forms in $P$ with $E_{A}$-valued
forms in $M$, let us see how to transfer the latter to classical
$\mathbb{R}$-valued forms in $M$. This is done by taking some ``trace''.
A \emph{trace} compatible with a FABS is given by a functor $\operatorname{tr}:\mathbb{Z}\mathbf{Alg}_{R}\rightarrow\mathbb{Z}\mathbf{Alg}_{\mathbb{R}}$
together with a natural transformation $\tau$ between $\operatorname{tr}\circ\Lambda(M;E_{-})$
and the constant functor at $\Lambda(M;\mathbb{R})$. All this data
fits in the following diagram:$$
\xymatrix{ \mathbf{Alg _{R}} \ar@{-->}[rr]^{E_{-}} && \mathbf{Alg _{R} Bun}_{M} \ar[d]^{\Lambda (M;-)} \ar@{<==}[ld]+<18pt>_-{\jmath} \\
\ar[u]^{\pi _1}  \mathbf{C} \ar@{-->}@/^{0.4cm}/[rr]^(.4){S(-,-)} \ar@/_{0.4cm}/[rr]_-{^{\Lambda (-;-)}} & \ar@{==>}[]+<0.8cm,-0.5cm><0.3cm>_{\imath} &\mathbb{Z} \mathbf{Alg _{R}}}
$$ 

Such systems exist, as shown in the next example. The fundamental
properties and constructions involving FABS will appear in a work
under preparation \citep{FABS}. 
\begin{example}[\emph{standard case}]
Let $\mathbf{C}$ be composed by pairs $(P,A)$, where $P\rightarrow M$
is a $G$-bundle whose group $G$ becomes endowed with a representation
$\rho:G\rightarrow GL(A)$. In that case we define $E_{-}$ as the
rule assigning to each $A$ the corresponding associated bundle $P\times_{\rho}A$.
The functor $S$ is such that $S(P,A)$ is the algebra $\Lambda_{\rho}(P;A)$
of \emph{$\rho$-equivariant $A$-valued forms $\alpha$ in }$P$,
i.e, of those which satisfy the equation $R_{g}^{*}\alpha=\rho(g^{-1})\cdot\alpha$,
where $R:G\times P\rightarrow P$ denotes the canonical free action
characterizing $P$ as a principal $G$-bundle. This algebra of $\rho$-equivariant
forms naturally embeds into $\Lambda(P;A)$, giving $\imath$. Finally,
it is a standard fact \citep{kobayashi} that each $\rho$-equivariant
$A$-valued form on $P$ induces an $P\times_{\rho}A$-valued form
on $M$, defining the transformation $\xi$. This is the standard
approach used in the literature, so that we will refer to it as the
\emph{standard FABS}. 
\end{example}

In order to define tghe EHP action functional we need to consider
invariant FABS. Let us introduce them. Indeed, we say that a FABS
on $\mathbf{C}$ compatible with a trace $(\operatorname{tr},\tau)$
is \emph{invariant }under a functor $I:\mathbf{C}\rightarrow\mathbb{Z}\mathbf{Alg}_{R}$
if
\begin{enumerate}
\item[(a)] for all $(P,A)\in\mathbf{C}$ the corresponding $I(P,A)$ is an ideal
of $S(P;A)$, so that we can take the quotient functor $S/I$ and
we have a natural transformation $\pi:S\Rightarrow S/I$;
\item[(b)] there exists another $J:\mathbf{C}\rightarrow\mathbb{Z}\mathbf{Alg}_{\mathbb{R}}$
such that $J(P,A)$ is an ideal of $\Lambda(M;E_{A})$ and whose projection
we denote by $\pi'$;
\item[(c)] there exists a natural transformation $\jmath':S/I\Rightarrow\Lambda(M;-)/J$
such that $\jmath'\circ\pi=\pi'\circ\jmath$, i.e, the diagram below
commutes for every $(P,A)$;\begin{equation}{ \label{invariant_FABS}
\xymatrix{\ar[d]_-{\pi _{(P,A)}} S(P,A)\ar[r]^-{\jmath _{(P,A)}} & \Lambda (M;E_A) \ar[d]^-{\pi ' _{(P,A)}} \\
S(P,A)/I(P,A) \ar[r]_{\jmath ' _{(P,A)}} & \Lambda (M;E_A)/J(P,A)}}
\end{equation}
\item[(d)]  the transformation $\tau$ passes to the quotient by $\tau\circ I$.
\end{enumerate}
Finally, notice that all these definitions make perfect sense in the
more general category $\mathfrak{m}\mathbf{Alg}_{R}$ of $\mathfrak{m}$-graded
$R$-algebras.

\section*{Acknowledgments}

Yuri X. Martins was supported by Capes. Both authors would like to
thank Romero Solha and Helvecio Fargnoli for fruitful discussions
and criticism of earlier drafts of the text.


\section*{References}
\begin{thebibliography}{99}
\bibitem{generalized_Einstein} Akbar-Zadeh, H.,\emph{ Generalized
Einstein manifolds,} J. of Geometry and Physics, v. 17, i. 4, 1995,
p. 342-380

\bibitem{ashteckar} Ashtekar, A.\emph{,} \emph{New Variables for
Classical and Quantum Gravity}. Physical Review Letters. 57 (18):
(1986), p. 2244\textendash 2247 

\bibitem{quantum_gravity_2} Baez, J. C., \emph{Degenerate solutions
of general relativity from topological field theory}, Comm. Math.
Phys. 193 (1998), no. 1, 219\textendash 231.

\bibitem{baez_knots} Baez, J. C., Muniain, J. P., Gauge Fields, Knots
and Gravity, World Scientific, 1994.

\bibitem{maximal_tori} Bayer-Fluckiger, E., \emph{Embeddings of maximal
tori in orthogonal groups}, Annales de l'institut Fourier, 64 no.
1 (2014), p. 113-125

\bibitem{berger_original} Berger M., \emph{Sur les groupes d'holonomie
homogènes de variétés à conexion affine et des variétés riemanniennes},
Bull. Soc. Math. France 283 (1955), 279\textendash 330.

\bibitem{einstein_manifolds_BESSE} Besse, A. L., \emph{Einstein Manifolds},
Berlin Heidelberg, 1987

\bibitem{symplectic_curvature_1} Bourgeois, F., Cahen, M., \emph{A
variational principle for symplectic connections}, J. of Geometry
and Physics, V. 30, I. 3, 1999, P. 233-265 

\bibitem{symplectic_curvature_2} Cahen, M., Gutt, S, Rawnsley, J.,
\emph{Symplectic connections with parallel Ricci tensor}, Banach Center
Publications 51.1 (2000): 31-41

\bibitem{quasi_einstein} Chaki, M. C., Maity, R. K., \emph{On quasi
Einstein manifolds}, Publ. Math. Debrecen 57 (2000), 297-306.

\bibitem{generalized_quasi_Einstein}Chaki, M. C., \emph{On generalized
quasi-Einstein manifolds}, Publ. Math. Debrecen, 58 (2001), 683\textendash 691

\bibitem{generalized_geometry_TYPE_II} Coimbra, A., Strickland-Constable,
C. and Waldram D., \emph{Supergravity as Generalised Geometry I: Type
II Theories}, arXiv:1107.1733.

\bibitem{generalized_quaternionic_1} Deschamps, G., \emph{Twistor
space of a generalized quaternionic manifold}, arXiv:1601.03882

\bibitem{complex_GR} Esposito, G., \emph{Complex General Relativity},
Springer Netherlands, 2002.

\bibitem{generalized _complex_thesis} Gualtieri, M, \emph{Generalized
Complex Geometry}, arXiv:math/0401221

\bibitem{hawking_ellis} Hawking, S., G. Ellis, \emph{The Large Scale
Structure of Space-Time}, Cambridge University Press, 1973 

\bibitem{generalized_complex_PRIMEIRO} Hitchin, N., \emph{Generalized
Calabi\textendash Yau Manifolds}, The Quarterly Journal of Mathematics,
Volume 54, Issue 3, 1 September 2003, Pages 281\textendash 308

\bibitem{generalized_complex_2} Hitchin, N., \emph{Lectures on generalized
geometry}, arXiv:1008.0973

\bibitem{quantum_gravity_1} Horowitz, G. T., \emph{Topology change
in classical and quantum gravity}, 1991, Class. Quantum Grav. 8 587

\bibitem{generalized_geometry_TYPE_1} Hull, C. M., \emph{Generalised
Geometry for M-Theory}, arXiv:hep-th/0701203

\bibitem{kobayashi}Kobayashi, S.,Nomizu, K., \emph{Foundations of
Differential Geometry, Vol.1}, Wiley-Interscience, 1996.

\bibitem{FABS} Martins, Y., X., Biezuner, R. J., \emph{Functorial
Algebra Bundle Systems}, in preparation.

\bibitem{narain_group} Narain, K. S., Sarmadi, M. H., Witten, E.,
\emph{A note on toroidal compactification of heterotic string theory},
Nuclear Physics B, Volume 279, Issues 3\textendash 4, 12 January 1987,
Pages 369-379

\bibitem{super_quasi_generalized} Özgür, C., \emph{On some classes
of super quasi-Einstein manifolds}, Chaos, Solitons and Fractals,
40 (2009) 1156\textendash 1161

\bibitem{more_subgroups_O(n)} Obata, M., \emph{On Subgroups of the
Orthogonal Group}, Transactions of the American Mathematical Society,
Vol. 87, No. 2 (Mar., 1958), pp. 347-358

\bibitem{generalized_quaternionic_INITIAL} Pantilie, R., \emph{Generalized
Quaternionic Manifolds}, arXiv:1109.6475

\bibitem{quaternionic_Kahler_SALAMON} Salamon, S., \emph{Quaternionic
K\"{a}hler Manifolds}, Inventiones mathematicae 67 (1982): 143-172

\bibitem{review_einstein}Sambusetti A., \emph{Einstein manifolds
and obstructions to the existence of Einstein metrics}, Recent developments
on Differential Geometry, Rendiconti di Matematica, serie VII, Vol.
18, 131-149 (1998)

\bibitem{Cayley_Dickson_2} Schafer, R. D., \emph{On the algebras
formed by the Cayley-Dickson process}, American Journal of Mathematics
76: 435\textendash 46 (1954)

\bibitem{exotic_nil_1} Smoktunowicz, A., \emph{Infinitely dimensional,
affine nil algebras} $A\otimes A^{op}$ \emph{and} $A\otimes A$ \emph{exist},
arXiv:1403.2557.

\bibitem{connections_graded_principal_bundles} Stavracou, T., \emph{Theory
of connections on graded principal bundles}, arXiv:dg-ga/9605006.

\bibitem{wald} Wald, R., \emph{General Relativity}, University of
Chicago Press, 1984

\bibitem{PhD_cartan_connections}Wise, D. K., \emph{Topological Gauge
Theory, Cartan Geometry, and Gravity}, PhD thesis, UCLA, Riverside.
\end{thebibliography}
\end{document}